\documentclass[11pt,a4paper,twocolumn]{article}

\usepackage{authblk}

\usepackage[T1]{fontenc}
\usepackage[utf8]{inputenc}
\usepackage{microtype}
\usepackage{geometry}
\usepackage{amsmath,amssymb,amsfonts,bm,amsthm}
\usepackage{dsfont}
\usepackage{physics}
\usepackage{graphicx}
\usepackage{booktabs}
\usepackage{hyperref}
\usepackage{xcolor}
\usepackage{enumitem}
\usepackage{mathtools}
\usepackage{subcaption}

\usepackage[style=phys, articletitle=true, biblabel=brackets, backend=biber]{biblatex}
\usepackage{libertinus}

\hypersetup{
colorlinks=true,
linkcolor=blue!60!black,
citecolor=blue!60!black,
urlcolor=blue!70!black
}

\newtheorem{Definition}{Definition}
\newtheorem{Proposition}{Proposition}

\newtheorem{Theorem}{Theorem}
\newtheorem*{theoremcopy}{Theorem}
\newtheorem{Lemma}{Lemma}

\newcommand{\Herm}[1]{\textnormal{Herm}(#1)}

\title{\textbf{The statistical disturbance bound of quantum measurements}}

\author[1,2]{Ties-A. Ohst}
\author[1,2]{Sebastian Schlösser}
\author[1,2]{Roope Uola}

\affil[1]{Department of Physics and Astronomy, Uppsala University, 75120 Uppsala, Sweden}
\affil[2]{Nordita, KTH Royal Institute of Technology and Stockholm University, 10691 Stockholm, Sweden}

\date{}

\begin{document}

\twocolumn[
\maketitle

\vspace{-1cm}

Quantifying the disturbance caused by a quantum measurement typically requires detailed knowledge of the underlying measurement channel.
In this work, we introduce a statistical disturbance bound, which connects the statistical properties of a quantum measurement to the state disturbance induced by any compatible measurement channel. Specifically, we show that the average fidelity between input and output with respect to an arbitrary ensemble of pure input states is fundamentally bounded in terms of the measurement, described as a positive operator-valued measure (POVM). We further develop the weighted state exclusion technique, which enables an experimental determination of the statistical disturbance bound without requiring explicit knowledge of the measurement effects.
To see the advantages of our approach over existing information-disturbance relations, we show that our bound distinguishes between measurements with equivalent informativeness. Furthermore, we demonstrate that the weighted state exclusion technique can detect and quantify measurement-induced disturbance using state preparations that are insufficient for tomographic reconstruction of the measurement operators. Finally, we illustrate how disturbance bounds defined with respect to specific input ensembles can be used to bound an eavesdropper's guessing probability in a simple protocol for quantum randomness generation.

\vspace{1cm}
]

\section{Introduction}
Among the most remarkable features of quantum theory is the unavoidable change that a physical system undergoes when information is extracted through a measurement. This insight dates back to Heisenberg's seminal microscope thought experiment \cite{Heisenberg1927, Heisenberg1949}. With the advent of quantum communication, where classical information is encoded into quantum systems to guarantee security against eavesdropping, measurement disturbance has become a central concept in quantum information science. In particular, the security of quantum key distribution protocols \cite{Scarani2009}, most notably the BB84 protocol \cite{Bennett2014}, relies fundamentally on the disturbance inevitably caused by an adversary's measurements. This connection has motivated extensive research on quantitative information-disturbance \cite{fuchs1996, Fuchs2001, Maccone2006, Maccone2007, Buscemi2008, Kretschmann2008} and error-disturbance relations \cite{ Ozawa2003, Ozawa2005,  Busch2013, heinosaari_miyadera2013}.

A particularly fruitful approach to quantifying measurement disturbance is through the average fidelity of the measurement channel \cite{schumacher1996}. Using this figure of merit, Banaszek derived a tight trade-off between information gain and state disturbance in Ref.~\cite{banaszek2001}, inspiring a number of subsequent developments \cite{cheong2012, Shitara2016, Lee2021, terashima2025}. One notable feature of Banaszek's relation is that it allows measurement-induced disturbance to be detected and quantified solely from the observed measurement statistics, without requiring access to the post-measurement states or an explicit characterisation of the measurement channel. However, these information-disturbance relations characterise disturbance only through the informativeness of the measurement and therefore do not address a more fundamental question: to what extent does the complete statistical description of a measurement constrain the disturbance that any compatible measurement channel must induce?

In this work, we answer this question by introducing the statistical disturbance bound, which provides a tight upper bound on the average fidelity, with respect to an arbitrary ensemble of pure input states, achievable among measurement channels that are compatible with a given positive operator-valued measure (POVM). The bound therefore establishes a direct link between the statistical description of a quantum measurement and the disturbance that it necessarily induces. When the POVM is known, the statistical disturbance bound can be computed efficiently by solving a semidefinite program. Moreover, we introduce the weighted state exclusion technique, illustrated in Fig.~\ref{fig:weighted_state_exclusion}, which enables an experimental determination of the bound without requiring explicit knowledge of the measurement effects.

\begin{figure}
    \centering
    \includegraphics[width=\linewidth]{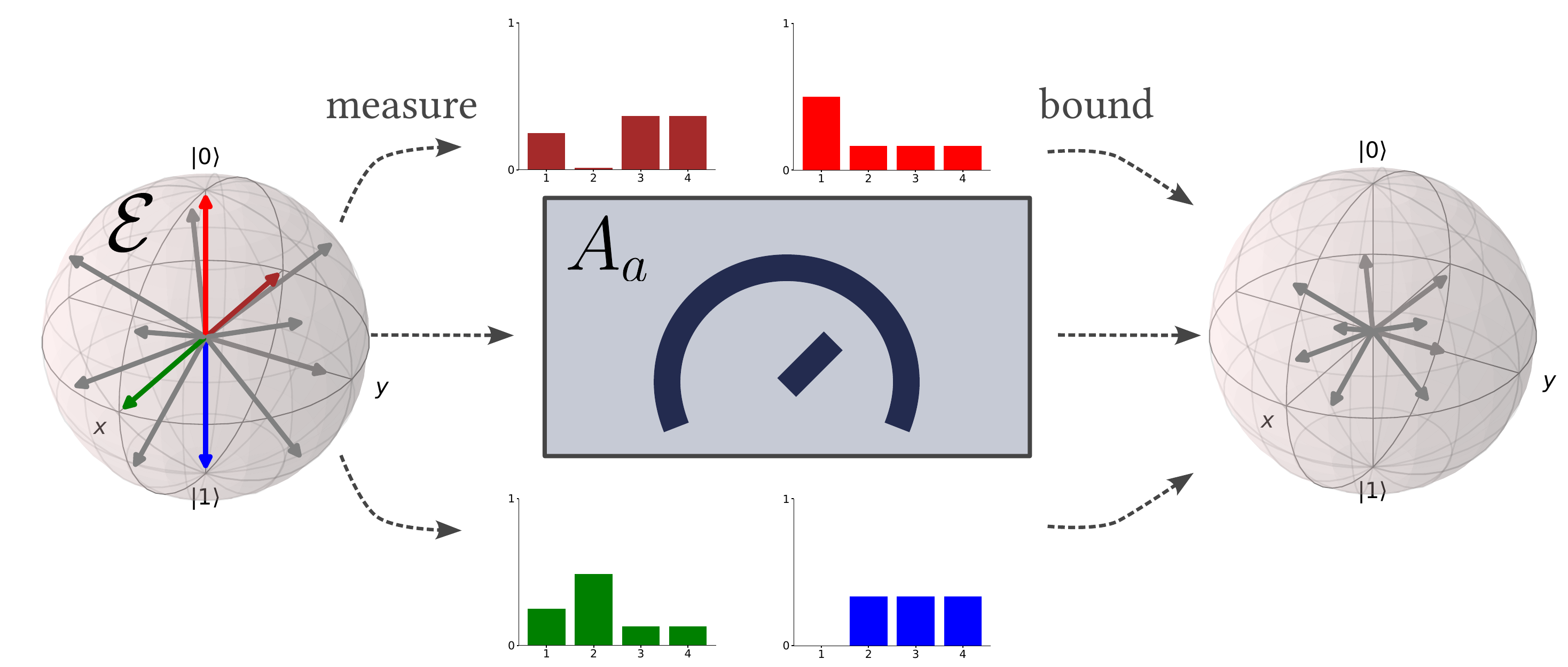}
    \caption{Sketch of the statistical disturbance estimation of a measurement with uncharacterised measurement effects $A_a$. A sufficiently large set of probe states $\rho_i$ (colored arrows) is prepared and measured to determine the corresponding conditional probabilities $p(a|i)$ (colored bar charts). These measurement statistics are then processed using the weighted state exclusion technique to obtain an upper bound on the maximum average fidelity achievable by any compatible measurement channel with respect to the input ensemble $\mathcal{E}$.}
    \label{fig:weighted_state_exclusion}
\end{figure}

The paper is organised as follows. In Section~\ref{sec:stat_dis_bound}, we review the elements of quantum measurement theory required to formulate the statistical disturbance bound. Section~\ref{sec:computation} presents its efficient computation via semidefinite programming when the POVM is known. For Haar-distributed input states, we derive a closed-form expression that reveals a direct connection between the ranks of the measurement operators and the corresponding measurement disturbance.

In Section~\ref{sec:measurement_of_sdb}, we turn to the experimentally relevant scenario in which the measurement device is not fully characterised. After reviewing in Section~\ref{sec:info_disturbance} how existing information-disturbance relations can be used for disturbance estimation, we introduce the weighted state exclusion technique in Section~\ref{sec:weighted_state_exclusion} as a significantly more powerful alternative. Since the method relies on trusted state preparations, Section~\ref{sec:preparation_uncertainties} shows how the resulting bounds can be made robust against preparation uncertainties.

Section~\ref{sec:applications} illustrates the advantages of our approach through three representative examples. In Section~\ref{sec:loss_vs_noise}, we demonstrate that the statistical disturbance bound distinguishes between information-equivalent measurements affected by depolarising noise and those subject to particle loss, a distinction that conventional information-disturbance relations fail to capture. In Section~\ref{sec:sic}, we explicitly show that the weighted state exclusion technique can detect measurement disturbance without requiring a tomographic reconstruction of the POVM. In Section~\ref{sec:amplitude_damping}, we investigate non-Haar input ensembles and show that the Lüders instrument does not always give rise to the least disturbing measurement channel compatible with a given POVM. 

Finally, in Section~\ref{sec:randomness}, we establish a connection between the average fidelity with respect to a particular input ensemble and a simple quantum randomness generation protocol, where the average fidelity directly yields an upper bound on an eavesdropper's guessing probability.

\section{Measurement disturbance and the statistical disturbance bound}
\label{sec:stat_dis_bound}

The disturbance caused by a quantum measurement refers to the extent to which the measurement alters the state of the measured system. Mathematically, the dynamics of a quantum measurement with $m\in\mathbb{N}$ outcomes acting on a $d$-dimensional system are described by an \textit{instrument}, that is, a collection $\bm{\mathcal{I}}=\{\mathcal{I}_a\}_{a=0}^{m-1}$ of \textit{completely positive} maps
$\mathcal{I}_a:\Herm{d}\rightarrow\Herm{d'}$
such that the map
$\mathcal{C}=\sum_a\mathcal{I}_a$
is \textit{trace-preserving}, i.e., the maps sum to a proper \textit{quantum channel}. Here, $\Herm{d}$ denotes the space of \textit{Hermitian} operators acting on $\mathbb{C}^d$.

Since our focus is on the change of quantum states, represented by \textit{positive semidefinite} operators $\rho\in\Herm{d}$ satisfying $\Tr(\rho)=1$, we restrict our attention to dimension-preserving instruments and therefore assume $d'=d$. Given an input state $\rho$, the corresponding non-selective state update is
\begin{equation*}
    \rho\mapsto\sum_a\mathcal{I}_a(\rho)=\mathcal{C}(\rho).
\end{equation*}
Whenever $\mathcal{C}(\rho)\neq\rho$, we say that the state $\rho$ is disturbed by the measurement $\bm{\mathcal{I}}$.

Rather than considering disturbance for a single input state, it is often more natural to study disturbance with respect to an ensemble of quantum states. Such an ensemble $\mathcal{E}$ is specified by a probability measure $\mu_{\mathcal{E}}(\psi)$ over the set of normalised pure states $\ket{\psi}\in\mathbb{C}^d$ satisfying $\braket{\psi}=1$. A central example considered throughout this work is the \textit{Haar ensemble} $\mathcal{H}$, which is induced by the normalised Haar measure on the unitary group $U(d)$; see Ref.~\cite{Mele2024} for an introduction to the Haar measure in quantum information theory. Intuitively, the Haar ensemble describes the absence of any prior knowledge about the prepared state and therefore assigns equal weight to all pure states. Another important class consists of finite ensembles, specified by pure states $\{\ket{\psi_i}\}_{i=1}^{n}$ with prior probabilities $\{p_i\}_{i=1}^{n}$. In this case,
$\mu_{\mathcal{E}}(\psi)=\sum_{i=1}^{n}p_i\delta(\psi-\psi_i)$,
where $\delta$ denotes the Dirac delta measure.

To quantify the disturbance induced by an instrument $\bm{\mathcal{I}}=\{\mathcal{I}_a\}_{a=0}^{m-1}$ with measurement channel $\mathcal{C}=\sum_a\mathcal{I}_a$, the \textit{$\mathcal{E}$-average fidelity} \cite{nielsen00} is defined as
\begin{equation}
    \label{eq:average_fidelity_integral}
    F_{\mathcal{E}}(\mathcal{C}) = \int d \mu_{\mathcal{E}}(\psi) \bra{\psi} \mathcal{C}(\ketbra{\psi}) \ket{\psi}
\end{equation}
where the integral is performed with respect to the probability measure $\mu_{\mathcal{E}}$.
The quantity $F_{\mathcal{E}}(\mathcal{C})$ measures the average overlap between the input state and the corresponding post-measurement state. It therefore quantifies the average state preservation with respect to the ensemble $\mathcal{E}$, while smaller values of $F_{\mathcal{E}}(\mathcal{C})$ indicate stronger measurement-induced disturbance.

Besides their dynamical description in terms of instruments, quantum measurements are commonly characterised by their statistical properties. From this perspective, a measurement is described by a \textit{positive operator-valued measure} (POVM), namely a collection of positive semidefinite operators $\bm{A}=\{A_a\}_{a=0}^{m-1}\subset\Herm{d}$ satisfying
$\sum_{a}A_a=\mathds{1}_d$.
For an input state $\rho$, the probability of obtaining outcome $a$ is given by the \textit{Born rule}
\begin{equation*}
    p(a|\rho) = \Tr(A_a \rho).
\end{equation*}
The statistical description provided by a POVM is strictly weaker than the dynamical description provided by an instrument. Indeed, every instrument determines a unique POVM through
\begin{equation}
    \label{eq:instrument_to_povm}
    A_a = \mathcal{I}_{a}^{\dagger}(\mathds{1}_d)
\end{equation}
where $\mathcal{I}_{a}^{\dagger}$ is the adjoint of the map $\mathcal{I}_{a}$ with respect to the Hilbert-Schmidt inner product.
Equation~\eqref{eq:instrument_to_povm} therefore allows the statistical description of a measurement to be inferred from its dynamics. The converse, however, is not true: a POVM does not uniquely determine a compatible instrument. In fact, every instrument $\bm{\mathcal{I}}$ compatible with a given POVM $\bm{A}$ via Eq.~\eqref{eq:instrument_to_povm} can be expanded in terms of the \textit{generalised Lüders rule}, see for instance \cite{Hayashi2006},
\begin{equation}
    \label{eq:gen_lüders_rule}
    \mathcal{I}_{a}(\rho) = \mathcal{C}_{a}\left(\sqrt{A_a} \rho \sqrt{A_a}\right),
\end{equation}
where $\mathcal{C}_a$ is an arbitrary quantum channel that may depend on the measurement outcome $a$. When $\mathcal{C}_a=\mathrm{id}_d$ for every outcome, the resulting instrument is known as the \textit{Lüders instrument} \cite{Busch1998}.

Our goal is to understand how the statistical description of a measurement constrains the disturbance that any compatible dynamical measurement implementation, i.e., an instrument of the form \eqref{eq:gen_lüders_rule}, must induce. More specifically, we seek fundamental upper bounds on the average fidelity $F_{\mathcal{E}}(\mathcal{C})$ that depend only on the POVM $\bm{A}$. Equivalently, given a measurement $\bm{A}$, we ask for the compatible instrument that maximises the average fidelity with respect to an ensemble $\mathcal{E}$, or, in other words, minimises the average disturbance.

A qualitative version of this question has previously been investigated in Ref.~\cite{heinosaari_miyadera2013} using ordering relations between quantum measurements and quantum channels. Here we develop a quantitative approach by introducing the \textit{statistical disturbance bound}.

\begin{Definition}
Let $\bm{A} = \{A_{a}\}_{a=0}^{m-1} \subset \Herm{d}$ be a POVM. Then, its statsitical disturbance bound $F_{\mathcal{E}}(\bm{A})$ with respect to the ensemble of pure states, specified by a probability measure $\mu_{\mathcal{E}}$, is defined as the solution of the following optimisation problem.  
    \begin{align}
    F_{\mathcal{E}}(\bm{A}) = \sup_{\bm{\mathcal{I}}} & \sum_{a=0}^{m-1}\int d\mu_{\mathcal{E}}(\psi) \bra{\psi} \mathcal{I}_{a}(\ketbra{\psi}) \ket{\psi} \label{eq:stat_dist_opt_def} \\
    \text{s.t. }  &\mathcal{I}_{a} \in \text{CP}(d,d), \mathcal{I}_{a}^{\dagger}(\mathds{1}_d) = A_a  \nonumber,
\end{align}
where $\text{CP}(d,d)$ denotes the set of completely positive maps between two $d$-dimensional quantum systems.
\end{Definition}

In words, $F_{\mathcal{E}}(\bm{A})$ is the largest average fidelity that can be achieved by any instrument compatible with the POVM $\bm{A}$. It therefore quantifies the minimum disturbance that is fundamentally compatible with the statistical description of the measurement described by the POVM $\bm{A}$.

In the next section, we show how this optimisation problem can be solved when the POVM $\bm{A}$ is known explicitly. We then explain how the statistical disturbance bound can be estimated experimentally without explicit knowledge of the measurement operators by means of the weighted state exclusion technique.

\section{Computation of the statistical disturbance bound}
\label{sec:computation}

The key insight underlying the computation of the statistical disturbance bound defined in Eq.~\eqref{eq:stat_dist_opt_def} is that the corresponding optimisation problem can be formulated as a semidefinite program (SDP). SDPs form a well-studied class of convex optimisation problems \cite{Boyd2004} with numerous applications in quantum information theory \cite{Watrous2018,Skrzypczyk2023}. They can be solved efficiently using numerical solvers \cite{Borchers1999CSDPAC} and, in some cases, even admit analytic solutions.

To derive an SDP formulation of Eq.~\eqref{eq:stat_dist_opt_def}, we employ the well-known Choi isomorphism \cite{Choi1972}, which associates every completely positive map $\mathcal{I}:\Herm{d} \rightarrow \Herm{d'}$ 
uniquely with a positive semidefinite operator $I \in  \Herm{d\cdot d'}$ via
\begin{equation*}
 \mathcal{I} \mapsto I = \frac{1}{d}\sum_{i,j=0}^{d-1} \ketbra{i}{j} \otimes \mathcal{I} (\ketbra{i}{j}).
\end{equation*}
With that, the optimisation problem \eqref{eq:stat_dist_opt_def} can be reformulated as a semidefinite program in standard form.
To do so, we introduce the \textit{ensemble operator} $R(\mathcal{E}) \in \Herm{d^2}$ defined by 
\begin{equation}
    \label{eq:ensemble_operator}
    R(\mathcal{E}) = d \int d\mu_{\mathcal{E}}(\psi) \ketbra{\psi}^T \otimes \ketbra{\psi},
\end{equation}
associated with any ensemble of pure states $\mathcal{E}$, where $X^T$ is the transpose of $X$. For instance, in the case of the $d$-dimensional Haar ensemble $\mathcal{H}$ one obtains from Ref.~\cite{horodeckis_singlet99} the ensemble operator 
\begin{equation}
     \label{eq:haar_ensemble_operator}
    R(\mathcal{H}) = \frac{\mathds{1}_{d^2} + d\ketbra{\phi^{+}}}{d+1}.
\end{equation}
with the maximally entangled state $\ket{\phi^{+}} = \frac{1}{\sqrt{d}} \sum_{i} \ket{i} \otimes \ket{i}$.
In terms of the ensemble operator, the average fidelity in Eq.~\eqref{eq:average_fidelity_integral} can be written as
\begin{equation}
    \label{eq:average_fidelity_choi}
    F_{\mathcal{E}}(\mathcal{I}) = \Tr(R(\mathcal{E}) I).
\end{equation}
Using the representation in Eq.~\eqref{eq:average_fidelity_choi}, we can reformulate the statistical disturbance bound via
\begin{align}
    F_{\mathcal{E}}(\bm{A}) = \sup_{I_a} &\sum_{a=0}^{m-1} \Tr(R(\mathcal{E})I_{a}) \label{eq:primal_fidelity_sdp}\\
    \text{s.t. }  &I_{a} \succcurlyeq 0, (\Tr_{\rm 2}(I_a))^T = A_{a}/d,  \nonumber
\end{align}
where we have used that the condition $\mathcal{I}^{\dagger}(\mathds{1}_d) = A_a$ is equivalent to $(\Tr_{\rm 2}(I))^T = A_a/d$ in terms of the Choi-operator $I$ of $\mathcal{I}$ and where $\Tr_2:\Herm{d\cdot d} \rightarrow \Herm{d}$ denotes the partial trace on the second factor. 
The optimisation problem in Eq.~\eqref{eq:primal_fidelity_sdp} is an SDP in standard form \cite{Watrous2018}, whose dual problem is
\begin{align}
     \inf_{Y_a} & \frac{1}{d}\sum_{a=0}^{m-1} \Tr(Y_{a} A_a) \label{eq:dual_fidelity_sdp} \\
    \text{s.t. }  &Y_a^T \otimes \mathds{1}_d \succcurlyeq R(\mathcal{E}). \label{eq:dual_constraints}
\end{align}
Weak duality of semidefinite programs \cite{Boyd2004} implies that every feasible choice of the dual variables $Y_a$ provides an upper bound 
\begin{equation*}
    \frac{1}{d}\sum_{a=0}^{m-1} \Tr(Y_{a} A_a)
\end{equation*}
on the optimal value of the primal problem \eqref{eq:primal_fidelity_sdp}. Furthermore, since $R(\mathcal{E})$ is positive semidefinite with trace $d$, the constraint
\begin{equation*}
   Y_a^T \otimes \mathds{1}_d \succcurlyeq R(\mathcal{E})
\end{equation*}
is strictly feasible, for example by choosing $Y_a=y_a\mathds1_d$ with $y_a>d$. Slater's condition \cite{Slater2013} therefore guarantees strong duality, implying that the optimal values of the primal and dual problems coincide.

The dual formulation not only enables efficient numerical computation, but also provides a powerful tool for obtaining analytic solutions. We illustrate this by deriving a closed-form expression for the Haar ensemble.
For that ensemble, using Eq.~\eqref{eq:haar_ensemble_operator}, the average fidelity in the case of the Lüders instrument $\mathcal{I}_a^{L}(\rho) = \sqrt{A_a} \rho \sqrt{A_a}$ is computed as 
\begin{align}
    \sum_{a=0}^{m-1} F_{\mathcal{H}}(\mathcal{I}_a^{L}) &= \sum_{a=0}^{m-1}\Tr(R(\mathcal{H}) I_a^{L}) \nonumber\\
    &= \sum_{a=0}^{m-1} \frac{\Tr(A_a) + \Tr(\sqrt{A_a})^2}{d(d+1)} \nonumber\\
    &= \frac{\frac{1}{d}\left[\sum_{a=0}^{m-1} \Tr(\sqrt{A_a})^2\right] + 1}{d+1}. \label{eq:haar_fidelity_of_meas}
\end{align}
In Ref.~\cite{banaszek2001}, Eq.~\eqref{eq:haar_fidelity_of_meas} was shown to be optimal among instruments whose elements each admit a single Kraus operator. We strengthen this result by proving that the same expression is optimal over all instruments compatible with the POVM $\bm A$.
\begin{Theorem}
    \label{thm:lüders_optimality}
    Let $\bm{A} = \{A_{a}\}_{a=0}^{m-1} \in \Herm{d}$ be a $d$-dimensional POVM. 
    Then, the statistical disturbance bound $F_{\mathcal{H}}(\bm{A})$ of $\bm{A}$ with respect to the Haar ensemble $\mathcal{H}$ is given by 
    \begin{equation}
        \label{eq:lüders_disturbance_bound}
        F_{\mathcal{H}}(\bm{A}) =  \frac{\frac{1}{d} \left[\sum_{a=0}^{m-1} \Tr (\sqrt{A_a})^2\right]+1}{d+1}.
    \end{equation}
    In particular, the disturbance is minimised by the Lüders instrument $\mathcal{I}_a^{L}(\rho) = \sqrt{A_a} \rho \sqrt{A_a}$.
\end{Theorem}
\begin{proof}
    The formula \eqref{eq:lüders_disturbance_bound} for $F_{\mathcal{H}}(\bm{A})$ follows from an analytic solution of the SDP \eqref{eq:primal_fidelity_sdp} by constructing dual feasible variables with respect to the constraints \eqref{eq:dual_constraints} that yields the above value \eqref{eq:lüders_disturbance_bound}. This technical construction is explicitly presented in Appendix \ref{app:lüders_optimality}.
\end{proof}
Theorem \ref{thm:lüders_optimality} provides an analytic formula for the statistical disturbance bound in the case of the Haar ensemble.
Furthermore, Theorem~\ref{thm:lüders_optimality} immediately implies that the Lüders instrument is minimally disturbing with respect to the Haar ensemble precisely when every nonzero POVM effect has rank one.
\begin{Proposition}
    \label{prop:rank_disturbance}
    Let $\bm{A} = \{A_{a}\}_{a=0}^{m-1} \in \Herm{d}$ be a $d$-dimensional POVM. Then it holds that 
    \begin{equation}
        \label{eq:statistical_disturbance_lower_bound}
        F_{\mathcal{H}}(\bm{A}) \geq \frac{2}{d+1}
    \end{equation}
    with equality if and only if $\text{rank}(A_a) \leq 1$ for all $a \in \{0,\dots,m-1\}$.
    Furthermore, whenever $\text{rank}(A_a) \leq k$ for all $a \in \{0,\dots,m-1\}$, one has the upper bound 
    \begin{equation}
        \label{eq:rank_upper_bound}
        F_{\mathcal{H}}(\bm{A}) \leq \frac{k+1}{d+1}.
    \end{equation}
\end{Proposition}
\begin{proof}
    To show the first part, notice that the inequality 
    \begin{equation}
    \label{eq:trace_root_inequality}
        \Tr(\sqrt{A})^2 \geq \Tr(A)
    \end{equation}
    holds for any positive semidefinite operator because a square of a sum of positive numbers is always greater than or equal to the corresponding sum of squares of the numbers. Also, inequality \eqref{eq:trace_root_inequality} reduces to equality only if $A$ has at most one nonzero eigenvalue, i.e., $\text{rank}(A) \leq 1$. Applying the inequality \eqref{eq:trace_root_inequality} to each effect $A_a$ in Eq.~\eqref{eq:lüders_disturbance_bound} yields the lower bound \eqref{eq:statistical_disturbance_lower_bound}. 

    To show the upper bound \eqref{eq:rank_upper_bound}, assuming that $\text{rank}(A) \leq k$ we can apply the Cauchy-Schwarz inequality for the Hilbert-Schmidt inner product via
    \begin{align}
       &\Tr(\sqrt{A})^2 = \Tr(\sqrt{A} \cdot \Pi_{A})^2 \nonumber\\ 
       &\leq \Tr(A) \cdot \Tr(\Pi_{A}) \leq k \Tr(A), \nonumber
    \end{align}
    where $\Pi_{A}$ denotes the projector onto the support of $A$. Applying this inequality to each POVM element in Eq.~\eqref{eq:lüders_disturbance_bound} immediately yields Eq.~\eqref{eq:rank_upper_bound}.
\end{proof}
Proposition~\ref{prop:rank_disturbance} admits several relevant remarks. First, Eq.~\eqref{eq:rank_upper_bound} also follows from the fidelity criterion of Ref.~\cite{morelli2023}, which was originally derived to detect the Schmidt number of quantum states. We nevertheless include the above proof because of its remarkable simplicity as a direct consequence of Eq.~\eqref{eq:lüders_disturbance_bound}. 

Second, the lower bound in Eq.~\eqref{eq:statistical_disturbance_lower_bound} is not the trivial lower bound satisfied by arbitrary quantum channels.
As can be seen from Eq.~\eqref{eq:haar_ensemble_operator} and \eqref{eq:average_fidelity_choi}, the latter universal bound is $1/(d+1)$ and is achieved by any channel whose Choi operator is orthogonal to $\ket{\phi^+}$. Consequently, any channel with Haar-average fidelity below $2/(d+1)$ cannot, by Proposition~\ref{prop:rank_disturbance}, be represented using only positive semidefinite Kraus operators.

We have therefore obtained a particularly simple expression for the statistical disturbance bound in the case of Haar-distributed input states. Besides enabling efficient evaluation, Eq.~\eqref{eq:lüders_disturbance_bound} reveals a direct connection between the ranks of the POVM elements and the disturbance induced by the measurement. For general ensembles $\mathcal E$, however, no analogous closed-form expression is known, and the statistical disturbance bound must be computed numerically by solving the SDP in Eq.~\eqref{eq:primal_fidelity_sdp}. In the next section, we show how this bound can be estimated experimentally even when the POVM itself is unknown.

\section{Measurement of the statistical disturbance bound}
\label{sec:measurement_of_sdb}

In the previous section, we showed that the statistical disturbance bound can be computed by solving a semidefinite program whenever the POVM $\bm{A}$ describing the measurement is known. In principle, this makes the bound experimentally accessible by first reconstructing the POVM using detector tomography \cite{Feito2009} and then evaluating the optimisation problem in Eq.~\eqref{eq:primal_fidelity_sdp}, or equivalently Eq.~\eqref{eq:lüders_disturbance_bound} for Haar-distributed input states.

In practice, however, this approach is unsatisfactory. Finite sampling errors in the tomographic reconstruction of the POVM propagate through the nonlinear post-processing required to evaluate the statistical disturbance bound, leading to substantial uncertainties. Our goal in this section is therefore to develop methods that estimate the disturbance directly from experimentally accessible measurement statistics, thereby avoiding explicit detector tomography.

\subsection{Disturbance bounds from information-disturbance relations}
\label{sec:info_disturbance}

Existing approaches for estimating measurement disturbance without reconstructing the post-measurement states are based on information-disturbance relations \cite{banaszek2001, Maccone2007, Kretschmann2008,Shitara2016}. Their central idea is closely related to the uncertainty principle: measurements that extract a large amount of information about certain quantum states must necessarily disturb at least some of those states.

Along these lines, a particularly important and in certain sense complete relation has been derived by Banaszek in Ref.~\cite{banaszek2001} in the form of a balance between the so-called estimation and operation fidelites of a measurement.

The estimation fidelity $E(\bm A)$, also called the informativeness of the measurement, quantifies how well classical information encoded into quantum states can be recovered from a single measurement outcome.
In mathematical terms, this corresponds to the solution of the optimisation problem
\begin{align}
    E(\bm{A}) = &\sup_{\rho_{a}} \frac{1}{d} \sum_{a=0}^{m-1} \Tr(A_a \rho_a) \label{eq:informativeness_optimisation}\\
    &\;\,\text{s.t.} \; \rho_{a} \succcurlyeq 0, \Tr(\rho_a) = 1, \nonumber
\end{align}
where the prefactor $1/d$ is chosen to guarantee that the attainable values of $E(\bm{A})$ lie in the interval $[\frac{1}{d}, 1]$. A related notion of the informativeness of a measurement has been recently introduced in Ref.~\cite{skrzypczyk2019} in the context of resource theories of quantum measurements.
Since maximizing $\Tr(A\rho)$ over all states $\rho$ simply selects the largest eigenvalue of $A$, the optimisation problem in Eq.~\eqref{eq:informativeness_optimisation} admits the closed-form solution
\begin{align}
    \label{eq:informativeness_eigvals}
    E(\bm{A}) &= \frac{1}{d} \sum_{a=0}^{m-1} \lambda_{\rm{max}}(A_{\rm a}).
\end{align}
The operational interpretation of $E(\bm A)$ makes it straightforward to estimate experimentally.
The idea is to prepare enough samples of states $\rho_i$ in order to infer the statistics $p(a|i) = \Tr(A_a \rho_i)$. Ideally, the prepared states approximate the eigenvectors corresponding to the largest eigenvalues. The sum of coincidence probabilities $p(a|a)$ divided by $d$ is then a lower bound to the informativeness $E(\bm{A})$, which can be easily seen from the definition given by the optimisation problem \eqref{eq:informativeness_optimisation}. In addition to its experimental availability, it has been shown by Banaszek in Ref.~\cite{banaszek2001} that one may use the informativeness to derive a bound to the statistical disturbance with respect to the Haar ensemble.
\begin{Proposition}[Banaszek, 2001 \cite{banaszek2001}]
    \label{prop:bana_info_dist}
    Let $\bm{A} = \{A_a\}_{a=0}^{m-1}$ be a POVM on $\mathbb{C}^{d}$ and let $\mathcal{H}$ be the $d$-dimensional Haar ensemble. Then it holds that 
    \begin{equation}
        \label{eq:banaszek_inequality}
        F_{\mathcal{H}}(\bm{A}) \leq B_{d}(E(\bm{A}))
    \end{equation}
    where $E(\bm{A})$ is the informativeness of $\bm{A}$ and $B_d(x)$ is the function defined by 
    \begin{equation}
        \label{eq:banasezk_function}
        B_d(x) = \frac{1+\frac{1}{d}\left[\sqrt{dx}+\sqrt{(d-1)(d-dx)}\right]^2}{d+1}.
    \end{equation}
    Furthermore, the inequality \eqref{eq:banaszek_inequality} is saturated if and only if the eigenvalues $\lambda_{a,i} := \lambda_{i}(A_a)$ (arranged in decreasing order in $i$) with $ i \in \{0,\dots, d-1\}$ are such that the vectors $\bm{v}_i : = (\sqrt{\lambda_{0,i}}, \dots, \sqrt{\lambda_{m-1,i}}) \in \mathbb{R}^m$ collecting the square roots of all $i$'th eigenvalues are all pairwise linearly dependent and all the vectors corresponding to the non-leading eigenvalues have the same length, i.e., $\norm{\bm{v}_1}_2= \norm{\bm{v}_2}_2 = \dots = \norm{\bm{v}_{d-1}}_2$.
\end{Proposition}
    Ref.~\cite{banaszek2001} further showed that Proposition~\ref{prop:bana_info_dist} is tight. Specifically, for every $x\in[1/d,1]$, there exists a POVM $\bm{A}$ satisfying
    \begin{equation*}
        E(\bm{A}) = x,\;  F_{\mathcal{H}}(\bm{A}) = B_{d}(x)
    \end{equation*}
Consequently, Proposition~\ref{prop:bana_info_dist} provides the optimal disturbance bound obtainable from the informativeness alone.
    
    Next to that, what makes Proposition \ref{prop:bana_info_dist} practically useful is the fact that the function $B_d(x)$ in Eq.~\eqref{eq:banasezk_function} is monotonically decreasing for $x \in [\frac{1}{d}, 1]$. This implies that every experimentally obtained lower bound on $E(\bm{A})$ gives rise to an upper bound of $B_d(E(\bm{A}))$ and hence by the inequality \eqref{eq:banaszek_inequality} also to a valid upper bound of the statistical disturbance bound $F_{\mathcal{H}}(\bm{A})$ with respect to the Haar ensemble. Since it is practically possible to measure lower bounds to the informativeness by the reasoning described above, the information-disturbance relation by Banaszek provides a strong method to estimate the disturbance of an uncharacterised quantum measurement. 

    Despite its practical usefulness, this approach has two important limitations. First, it applies only to the Haar ensemble, restricting the class of disturbance phenomena that can be investigated. Second and more importantly, the bound in Eq.~\eqref{eq:banaszek_inequality} is generally far from tight for a given measurement, as we shall demonstrate in Section~\ref{sec:loss_vs_noise}. These shortcomings motivate the weighted state exclusion technique for a direct measurement of the statistical disturbance bound.

\subsection{Weighted quantum state exclusion}
\label{sec:weighted_state_exclusion}
The weighted state exclusion technique is based on a simple observation: the dual SDP in Eq.~\eqref{eq:dual_fidelity_sdp} admits a natural operational interpretation.
In fact, the optimisation variables $Y_a \in \Herm{d}$ in the problem \eqref{eq:dual_fidelity_sdp} may be identified as some unnormalised quantum state associated to the operator $A_a$.
Interpreting the dual variables as unnormalised quantum states reveals a close connection with the quantum state exclusion problem introduced in Ref.~\cite{Bandyopadhyay2014}.

To make this connection precise, suppose that there is a set of $n$ different states $\bm{\rho} = \{\rho_{i}\}_{i=1}^{n} \subset \Herm{d}$ that are prepared with high fidelity, e.g., the pure states of a preferred orthornormal basis, and measured with the uncharacterised POVM $\bm{A}$. After sufficiently many repetitions, we can estimate the set of conditional probability distributions $p(a|i) = \Tr(A_a \rho_i)$. If we restrict the dual variables to admit the decomposition
\begin{equation}
    \label{eq:dual_variable_state_deco}
    Y_a = \sum_{i=1}^{n} q_{ai} \, \rho_{i}
\end{equation}
with some expansion coefficients $q_{ai} \in \mathbb{R}$, we have also access to the values $\Tr(A_{a} Y_a)$ via the simple post-processing
\begin{equation}
    \label{eq:single_outcome_dual_objective_function}
    \Tr(A_{a} Y_a) = \sum_i q_{ai} \Tr(A_{a} \rho_{i}) = \sum_i q_{ai} \,p(a|i). 
\end{equation}
This means that we can obtain an upper bound to $F_{\mathcal{E}}(\bm{A})$ by considering the optimisation problem \eqref{eq:dual_fidelity_sdp} and invoking the additional restriction \eqref{eq:dual_variable_state_deco} posed on the optimisation variables $Y_a$. Since the numbers $q_{ai}$ are free variables, one can furthermore minimise the objective function \eqref{eq:single_outcome_dual_objective_function} over them,
and the resulting value will be an upper bound to $F_{\mathcal{E}}(\bm{A})$. 
For reasons that become apparent in the next subsection, we furthermore require that the expansion coefficients $q_{ai}$, also referred to as weights, are all non-negative. 
The previous discussion together with the form of the dual SDP \eqref{eq:dual_fidelity_sdp} directly imply the following Theorem.

\begin{Theorem}
    \label{thm:weighted_state_exclusion}
    Let $\bm{A} = \{A_{a}\}_{a=0}^{m-1} \subset \Herm{d}$ be a POVM, $\mathcal{E}$ an ensemble of pure states and let $ \bm{\rho} = \{\rho_{i}\}_{i=1}^{n} \subset \Herm{d}$ be a list of distinct quantum states. Then, the upper bound given by
    \begin{equation}
        \label{eq:weighted_state_exclusion_inequality}
        F_{\mathcal{E}}(\bm{A}, \bm{\rho}) \geq F_{\mathcal{E}}(\bm{A})
    \end{equation}
holds, where the value $F_{\mathcal{E}}(\bm{A}, \bm{\rho})$ is defined by  
    \begin{align}
      F_{\mathcal{E}}(\bm{A}, \bm{\rho}) := &\inf_{q_{ai}} \frac{1}{d}\sum_{a=0}^{m-1} \sum_{i=1}^{n} q_{ai} \Tr(A_{a} \rho_{i}) \label{eq:weighted_state_exclusion} \\
      &\, \textnormal{w.r.t. } q_{ai} \in \mathbb{R}_{+} \nonumber\\
      &\; \textnormal{s.t. } \sum_{i=1}^{n} q_{ai} \, \rho_{i}^{T} \otimes \mathds{1}_d \succcurlyeq R(\mathcal{E}) \;\, \forall a. \label{eq:ideal_dual_constraint}
    \end{align}
    Furthermore, with a suitable choice of states $\rho_i$ the inequality \eqref{eq:weighted_state_exclusion_inequality} can always be made tight.
\end{Theorem}
We call the procedure of determining $F_{\mathcal{E}}(\bm{A}, \bm{\rho})$, as explained above, the weighted state exclusion technique.
The name ``weighted state exclusion'' reflects the fact that the optimisation assigns non-negative weights to a collection of trusted probe states. These weighted states define feasible dual variables $Y_a$, whose expectation values can be estimated directly from experimentally observed statistics. 
In this way, an upper bound to the statistical disturbance bound is inferred without reconstructing the POVM.

The bound $F_{\mathcal{E}}(\bm{A}, \bm{\rho})$ in Theorem~\ref{thm:weighted_state_exclusion} can always be made tight by choosing the probe states to coincide with the eigenvectors of the optimal dual variables.
In this case one has
\begin{equation*}
    \bm{\rho} = \bigcup_{a=0}^{m-1} \{\ketbra{\psi_{ai}}\}_{i=1}^{d}
\end{equation*}
where $\ket{\psi_{ai}}$ denote the eigenvectors of the variables $Y_a$ that optimise the SDP \eqref{eq:dual_fidelity_sdp} and the variables $q_{ai}$ are the corresponding eigenvalues, i.e., $Y_a = \sum_{i=1}^{d} q_{ai} \ketbra{\psi_{ai}}$ with $\braket{\psi_{ai}}{\psi_{aj}} = \delta_{ij}$.

\subsection{Robustness to preparation uncertainties}
\label{sec:preparation_uncertainties}

The main practical limitation of the weighted state exclusion technique is its reliance on accurate state preparation of the probe states $\rho_i$. 
If the experimentally prepared state $\tilde\rho_i$ differs from the assumed ideal state $\rho_i$, the constraint in Eq.~\eqref{eq:ideal_dual_constraint} is generally no longer guaranteed to hold.
As a consequence, the value $F_{\mathcal{E}}(\bm{A}, \bm{\rho})$ computed based on the experimental data might fail to be a valid upper bound to $F_{\mathcal{E}}(\bm{A})$  potentially invalidating the resulting disturbance bound.

To overcome this issue, we will now show that the value $F_{\mathcal{E}}(\bm{A}, \bm{\rho})$ can be made robust in the sense that it gives a valid bound, even in the presence of preparation uncertainties that are quantified by non-vanishing trace distances $\norm{\rho - \Tilde{\rho}}_{1} := \Tr(\sqrt{(\rho - \Tilde{\rho})^2})$. Mathematically, we use the following lemma that takes into account a non-zero trace distance between the ideal and real state. 
\begin{Lemma}
    \label{lemma:robust_dual_inequality}
   Let $ \bm{\rho} = \{\rho_{i}\}_{i=1}^{n} \subset \Herm{d}$ and $ \Tilde{\bm{\rho}} = \{\Tilde{\rho}_{i}\}_{i=1}^{n} \subset \Herm{d}$ be two lists of quantum states such that $\norm{\rho_{i} - \Tilde{\rho}_{i}}_{1} \leq \epsilon$ and let $\{q_i\}_{i=1}^{n} \subset \mathbb{R}_{+}^{n}$ be a list of non-negative numbers. Let furthermore $\mathcal{E}$ be an ensemble of quantum states with associated ensemble operator $R(\mathcal{E})$.
    Then, if the inequality
    \begin{equation}
        \label{eq:relaxed_dual_constraint}
        \sum_{i=1}^{n} q_i \, \rho_{i} \otimes \mathds{1}_{d} \succcurlyeq R(\mathcal{E})+\epsilon \sum_{i=1}^{n}q_{i} \mathds{1}_{d^2}
    \end{equation}
    is satisfied, also the inequality 
    \begin{equation}
        \label{eq:real_state_dual_constraint}
        \sum_{i=1}^{n} q_i \, \Tilde{\rho}_{i} \otimes \mathds{1}_{d} \succcurlyeq R(\mathcal{E})
    \end{equation}
    holds true. 
\end{Lemma}
\begin{proof}
    First, notice that the condition $\norm{\rho_{i} - \Tilde{\rho}_{i}}_{1} \leq \epsilon$ implies that all eigenvalues of the operator $\rho_{i} - \Tilde{\rho}_{i}$ lie in the interval $[-\epsilon,\epsilon]$ implying that 
    \begin{equation}
        \label{eq:one_norm_to_inequality}
        \Tilde{\rho}_{i}  \succcurlyeq  \rho_{i} - \epsilon \mathds{1}_{d}.
    \end{equation}
    Next, we can rewrite the inequality \eqref{eq:relaxed_dual_constraint} as 
    \begin{equation*}
        \sum_{i=1}^{n} q_i \, (\rho_{i} - \epsilon \mathds{1}_{d}) \otimes \mathds{1}_{d} \succcurlyeq R(\mathcal{E}).
    \end{equation*}
    Together with the inequality \eqref{eq:one_norm_to_inequality}, this implies the inequality \eqref{eq:real_state_dual_constraint}.
\end{proof}

Lemma \ref{lemma:robust_dual_inequality} allows to formulate a stronger version of the weighted state exclusion technique \eqref{eq:weighted_state_exclusion} via the optimisation problem
\begin{align}
      & F_{\mathcal{E}}(\bm{A}, \bm{\rho}, \epsilon) :=\inf_{q_{ai}} \frac{1}{d}\sum_{a=0}^{m-1} \sum_{i=1}^{n} q_{ai} \Tr(A_{a} \rho_{i}) \label{eq:epsilon_weighted_state_exclusion} \\
      & \textnormal{w.r.t. } q_{ai} \in \mathbb{R}_{+} \nonumber\\
      & \textnormal{s.t. } \sum_{i=1}^{n} q_{ai} \, \rho_{i}^{T} \otimes \mathds{1}_d \succcurlyeq R(\mathcal{E}) +\epsilon \sum_{i=1}^{n}q_{ai} \mathds{1}_{d^2} \;\, \forall a. \nonumber
\end{align}
The modified optimisation problem therefore yields a disturbance bound that remains valid even when the prepared states deviate from their ideal descriptions by at most $\epsilon$ in trace distance.
This may be summarised by the inequality 
\begin{equation*}
    F_{\mathcal{E}}(\bm{A}, \bm{\rho}, \epsilon) \geq  F_{\mathcal{E}}(\bm{A}, \bm{\rho}) \geq F_{\mathcal{E}}(\bm{A}).
\end{equation*}
Consequently, whenever one has that $\norm{\rho_i - \Tilde{\rho}_i}_1 \leq \epsilon$, observing
\begin{equation*}
    F_{\mathcal{E}}(\bm{A}, \bm{\rho}, \epsilon) < 1
\end{equation*}
certifies that the measurement necessarily disturbs the states from the ensemble $\mathcal{E}$. In case of the optimal state preparations $\bm{\rho}$ such that $F_{\mathcal{E}}(\bm{A}, \bm{\rho}) = F_{\mathcal{E}}(\bm{A})$, we use the short-hand notation $ F_{\mathcal{E}}(\bm{A}, \epsilon)$ for the bound in the presence of preparation uncertainties.

\section{Examples}
\label{sec:applications}

In this section, we illustrate the usefulness of the statistical disturbance bound through several examples.

\subsection{Sample loss implies stronger disturbance than informativeness-equivalent depolarisation noise}
\label{sec:loss_vs_noise}

\begin{figure}
    \centering
    \includegraphics[width=\linewidth]{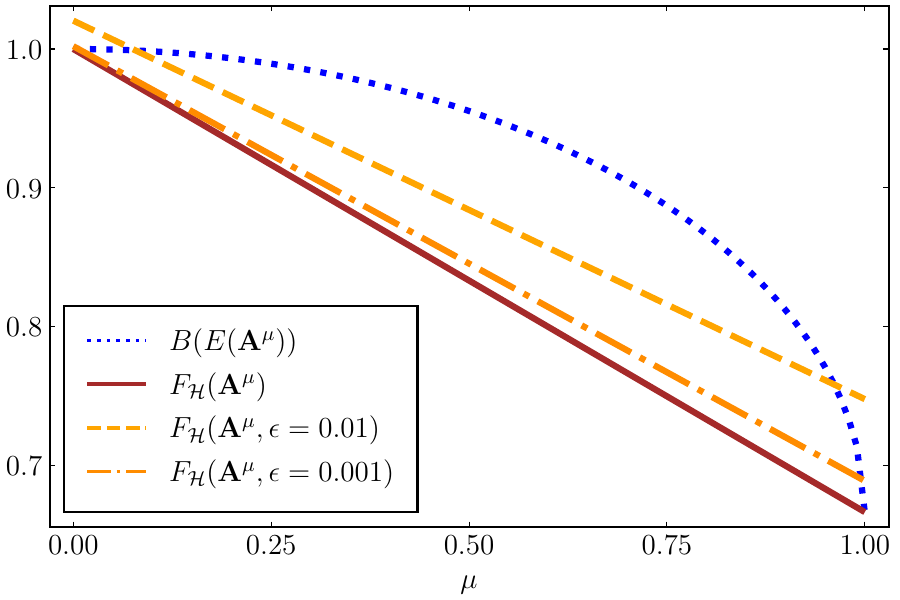}
    \caption{Comparison between depolarisation noise and sample loss for a qubit computational-basis measurement. The solid curve shows the statistical disturbance bound for a lossy measurement as a function of the transmission probability $\mu$. The dotted curve shows the information-disturbance bound $B_2(E(\bm A^\mu))$, which coincides with the statistical disturbance bound of the depolarised measurement with the same value of $\mu$. Thin solid curves illustrate the robust weighted state exclusion bounds for different preparation uncertainties $\epsilon$. Even for moderate preparation errors, the statistical disturbance bound remains significantly tighter than the information-disturbance relation over a broad parameter range.}
    \label{fig:lossy_plot}
\end{figure}

Our first example demonstrates that two measurements with identical informativeness may nevertheless induce substantially different disturbance. This illustrates a key advantage of the statistical disturbance bound over information-disturbance relations such as the one reviewed in Section~\ref{sec:info_disturbance}.

A particularly instructive example is provided by comparing depolarising noise with sample loss.
Consider the ideal projective measurement in the computational basis,
\begin{equation*}
    A_a = \ketbra{a}
\end{equation*}
for $a \in \{0,\dots,d-1\}$.
For depolarising noise with mixing parameter $\mu\in[0,1]$, the POVM becomes
\begin{equation*}
    \widetilde{A}_{a}^{\mu} = \mu \ketbra{a} + \frac{1-\mu}{d} \mathds{1}_d.
\end{equation*}
Using Eq.~\eqref{eq:informativeness_eigvals}, the informativeness is
\begin{equation}
    \label{eq:depolarisation_informativeness}
    E(\widetilde{\bm{A}}^{\mu}) = \frac{1 + \mu (d-1)}{d}.
\end{equation}
Since the vectors of ordered eigenvalues
$$\left(\mu + \frac{1-\mu}{d}, \frac{1-\mu}{d}, \dots, \frac{1-\mu}{d}\right)$$
are the same for each POVM element, and since the non-leading eigenvalues are all identical,
 Proposition~\ref{prop:bana_info_dist}  implies that the information-disturbance relation is saturated, i.e.,
$$B_d(E(\widetilde{\bm{A}}^{\mu})) = F_{\mathcal{H}}(\widetilde{\bm{A}}^{\mu}).$$
Hence, the informativeness completely determines the statistical disturbance bound, which may be computed as
\begin{align}
    \label{eq:depolarisation_stat_disturbance}
    F_{\mathcal{H}}(\widetilde{\bm{A}}^{\mu}) = \frac{\left(\sqrt{\mu + \frac{1-\mu}{d}} + (d-1) \sqrt{\frac{1-\mu}{d}}\right)^2 + 1}{d+1},
\end{align}
see Figure \ref{fig:lossy_plot}.

We now compare this with particle loss, corresponding to the situation where the sample state is lost before reaching the detector.
Equivalently, this model describes a detector with finite efficiency.
If the detector fails to click with probability $1-\mu$, the POVM is updated as $\bm{A} = \{A_a\}_{a=0}^{d-1} \mapsto \bm{A}^{\mu} =  \{A_{a}^{\mu}\}_{a=1}^{d}$ with
\begin{align}
    A_a^{\mu} &= \mu A_a \; \text{ for }  a \in \{0 \dots, d-1\} \nonumber\\
    A_{d} &= (1-\mu) \mathds{1}_d, \nonumber 
\end{align}
where the outcome $d$ refers to the ``no-click'' event. 

Both error models preserve exactly the same probability of correctly identifying the computational basis states.
Indeed, one readily verifies that the informativeness $E(\bm{A}^{\mu})$ of the lossy measurement coincides with the informativeness $E(\widetilde{\bm{A}}^{\mu})$ of the depolarised version of $\bm{A}$, see Eq.~\eqref{eq:depolarisation_informativeness}, i.e., $E(\bm{A}^{\mu}) = E(\widetilde{\bm{A}}^{\mu})$. 
Despite their identical informativeness, the two measurements exhibit different statistical disturbance bounds, as the lossy measurement satisfies
\begin{align}
    F_{\mathcal{H}}(\bm{A}^{\mu}) = 1 +  \frac{\mu (1-d)}{d+1}. \nonumber
\end{align}
Thus, the disturbance bound depends linearly on the loss parameter $\mu$, in contrast to the nonlinear dependence obtained for depolarising noise in Eq.~\eqref{eq:depolarisation_stat_disturbance}.
This difference is depicted in Figure~\ref{fig:lossy_plot}.
For every $\mu\in(0,1)$, the lossy measurement exhibits a strictly smaller average fidelity than its depolarised counterpart despite having exactly the same informativeness.
Physically, this difference originates from the fact that the informative outcomes of the lossy measurement correspond to sharp projective state updates while the depolarised measurement
keeps part of the coherence. 
Although both mechanisms reduce the informativeness by the same amount, they constrain the compatible measurement channels differently.

This example clearly demonstrates that the complete measurement statistics contain information about measurement disturbance that is not captured by informativeness alone. Consequently, the statistical disturbance bound can distinguish physically different noise mechanisms that remain indistinguishable within the conventional information-disturbance relation.

\subsection{Disturbance bounds from incomplete weighted state exclusion statistics}
\label{sec:sic}

To illustrate the application of the weighted state exclusion technique with incomplete and non-optimal state preparations, we consider the symmetric informationally complete (SIC) POVM \cite{Renes2004} for a qubit.
The POVM $\bm S$ consists of four rank-one effects $S_{a} = \frac{1}{2} \ketbra{\psi_a}$,
where the states ${\ket{\psi_a}}$ form the vertices of a regular tetrahedron on the Bloch sphere, e.g.,
\begin{align}
    \ket{\psi_0} &= \ket{0} \nonumber\\
    \ket{\psi_j} &= \frac{1}{\sqrt{3}}\ket{0} + \sqrt{\frac{2}{3}} e^{\frac{2 \pi i (j-2)}{3}}\ket{1}, j \in \{1,2,3\}. \nonumber
    \end{align}
Analogously to the depolarised measurement in the computational basis, we introduce a noisy SIC POVM $\bm{S}^{\mu}$ via the measurement effects $S_a^{\mu}$ defined as 
\begin{equation}
    \label{eq:noisy_sic}
    S_a^{\mu} = \mu S_a + \frac{1-\mu}{4} \mathds{1}_2
\end{equation}
with a noise parameter $\mu \in [0,1]$. 
By Theorem~\ref{thm:lüders_optimality}, the statistical disturbance bound with respect to the Haar ensemble coincides with that of the depolarised qubit basis measurement,
\begin{align}
    F_{\mathcal{H}}(\bm{S}^{\mu}) = \frac{\left(\sqrt{\mu + \frac{1-\mu}{2}} + \sqrt{\frac{1-\mu}{2}}\right)^2 + 1}{3}. \nonumber
\end{align}
Unlike the computational-basis measurement, the effects $S_a^\mu$ do not commute and therefore cannot be diagonalized simultaneously. Consequently, there is no single preparation basis that is naturally adapted to all measurement effects. This raises the question of whether weighted state exclusion can still detect measurement disturbance when only incomplete sets of probe states are available.

Figure~\ref{fig:sic_disturbance} demonstrates that the answer is affirmative. It shows the weighted state exclusion bounds
$F_{\mathcal{H}}(\bm{S}^{\mu}, \bm{\rho})$ for measurement disturbance obtained from the optimisation problem
\eqref{eq:weighted_state_exclusion}
for two different collections $\bm{\rho}$ of probe states. The set $\bm{\rho}_{Z}$ consists of the eigenstates of $\sigma_z$ while $\bm{\rho}_{ZX}$ consists of the union of both of the eigenstates of $\sigma_z$ and $\sigma_x$, i.e.,
\begin{align}
    \bm{\rho}_Z &= \{\ketbra{0}, \ketbra{1}\}, \nonumber \\
    \bm{\rho}_{ZX} &= \{\ketbra{0}, \ketbra{1}, \ketbra{+}, \ketbra{-}\}. \nonumber
\end{align}
Neither preparation set is informationally complete for detector tomography. In particular, reconstructing the noisy SIC POVM would additionally require measurements in a third basis, for example the eigenbasis of $\sigma_y$.
Nevertheless, Fig.~\ref{fig:sic_disturbance} shows that both preparation sets yield weighted state exclusion bounds that remain strictly below $1$ for sufficiently large values of $\mu$. Although these bounds are necessarily weaker than the exact statistical disturbance bound,
the values $F_{\mathcal{H}}(\bm{S}^{\mu}, \bm{\rho})$ remain significantly below $1$ for large enough $\mu$, thereby certifying that measurement disturbance can nevertheless be detected. This illustrates a key practical distinction between detector tomography and weighted state exclusion. 
The dashed curves illustrate the robust bounds obtained for preparation uncertainty $\epsilon=0.05$. Even in this case, the bounds remain below $1$ over a substantial range of $\mu$, demonstrating that disturbance can still be certified despite moderate state-preparation errors.

While tomography aims to reconstruct the complete POVM, disturbance certification only requires enough statistical information to construct a feasible dual solution of the optimisation problem \eqref{eq:dual_fidelity_sdp}. Consequently, as it is demonstrated by the above example, incomplete measurement statistics can already suffice.

\begin{figure}
    \centering
    \includegraphics[width=\linewidth]{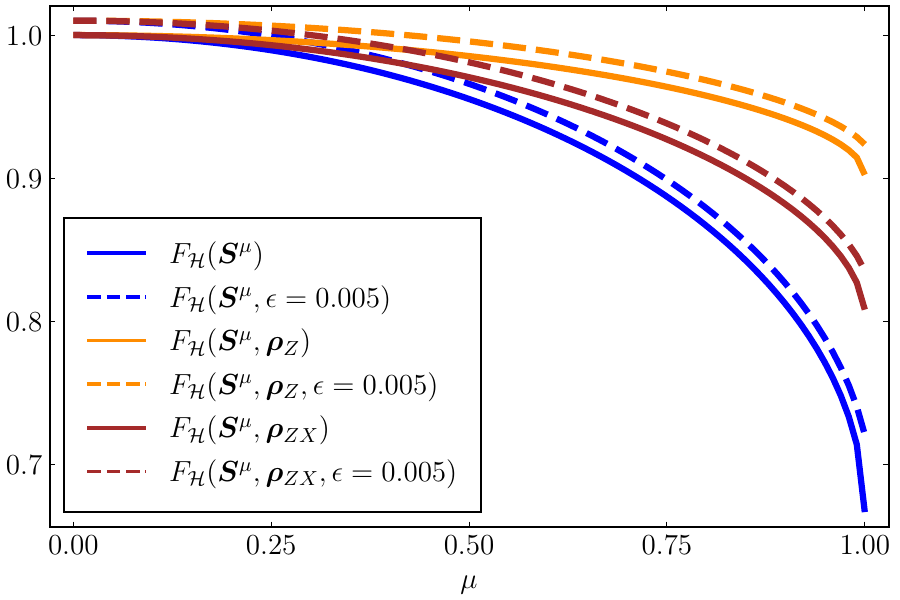}
    \caption{Weighted state exclusion with incomplete state preparations. The lowest solid curve shows the exact statistical disturbance bound for the noisy SIC POVM defined in Eq.~\eqref{eq:noisy_sic}. The upper solid curves show the experimentally accessible weighted state exclusion bounds obtained using only the preparation sets $\bm\rho_Z$ and $\bm\rho_{ZX}$. Neither preparation set is sufficient for detector tomography, yet both certify nonzero measurement disturbance. Dashed curves show the corresponding robust bounds for preparation uncertainty $\epsilon=0.005$.}
    \label{fig:sic_disturbance}
\end{figure}

\subsection{Non-optimality of the Lüders instrument for general state ensembles}
\label{sec:amplitude_damping}

So far, we have focused on measurement disturbance with respect to the Haar ensemble. We now show that considering different input ensembles reveals a qualitatively new behaviour.
As an illustration, we consider a measurement whose statistics arise from an amplitude damping channel followed by a computational-basis measurement. In particular, we demonstrate that the Lüders instrument is generally not the least disturbing instrument once the input ensemble differs from the Haar ensemble. 

For a damping parameter $\mu\in[0,1]$, the qubit amplitude damping channel is defined by $\mathcal{C}_{\mu}:\Herm{2} \rightarrow\Herm{2}$ acting as
\begin{equation*}
    \mathcal{C}_{\mu}[\rho] = K_1 \rho K_1^{\dagger} + K_2 \rho K_2^{\dagger}
\end{equation*}
where the Kraus operators $K_1$ and $K_2$ are 
\begin{align}
    K_1 = \begin{pmatrix}
        1 & 0 \\
        0 & \sqrt{\mu}
    \end{pmatrix}, \;
    K_2 = \begin{pmatrix}
        0 & \sqrt{1-\mu} \\
        0 & 0 
    \end{pmatrix}. \nonumber
\end{align}
The corresponding measurement is described by the POVM $\bm{B}^{\mu} = \{B_{0}^{\mu}, B_{1}^{\mu}\}$ whose effects are given by 
\begin{equation}
    \label{eq:amplitude_damping_measurement}
    B_{b}^{\mu} := \mathcal{C}_{\mu}^{\dagger}[\ketbra{b}],
\end{equation}
so that they explicitly take the form
\begin{align}
    B_0^{\mu} &= \begin{pmatrix}
        1 & 0 \\
        0 & 1-\mu
    \end{pmatrix}, \;
    B_1^{\mu} = \begin{pmatrix}
        0 & 0 \\
        0 & \mu 
    \end{pmatrix}. \nonumber
\end{align}
The associated Lüders channel $$\mathcal{B}_{L}^{\mu}(\rho) = \sum_{b} \sqrt{B_{b}^{\mu}} \rho \sqrt{B_{b}^{\mu}}$$ is therefore a partially dephasing channel 
\begin{equation}
    \label{eq:ad_lüders_channel}
    \mathcal{B}_{L}^{\mu}(\rho) = \begin{pmatrix}
        \rho_{00} & \sqrt{1-\mu} \rho_{01} \\
        \sqrt{1-\mu} \rho_{10} & \rho_{11}.
    \end{pmatrix}
\end{equation}
This channel preserves the projection on the $\sigma_z$-axis while attenuating the coherences by the factor $\sqrt{1-\mu}$. Consequently, the disturbance depends strongly on the latitude of the input state on the Bloch sphere. 
Although Theorem~\ref{thm:lüders_optimality} guarantees that the Lüders instrument is optimal for the Haar ensemble, we will show that it becomes suboptimal for most of the latitude ensembles.

To demonstrate this, we consider the latitude ensembles $\mathcal E_\theta$, consisting of uniformly distributed pure states with fixed polar angle $\theta$ on the Bloch sphere, see Fig.~\ref{fig:latitude_bloch_sphere}.
\begin{figure}
    \centering
    \includegraphics[width=0.6\linewidth]{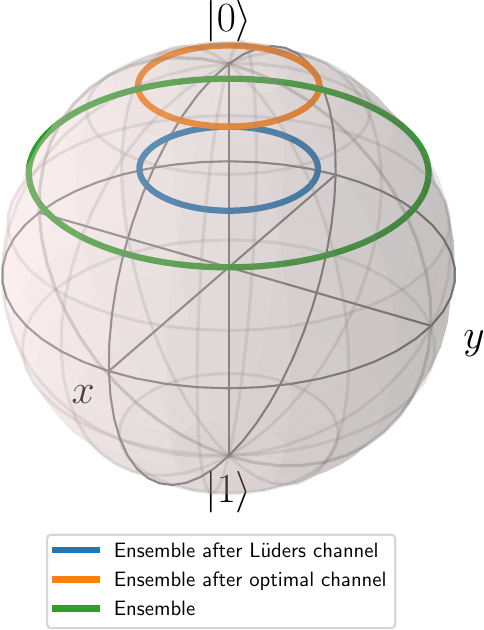}
    \caption{The latitude ensemble $\mathcal{E}_{\theta}$ for a polar angle $\theta$ in the northern hemisphere is represented by a ring on the Bloch sphere at fixed latitude. The Lüders instrument associated with the amplitude-damped measurement acts as a partially dephasing channel, reducing the radius of the ring while leaving its polar angle unchanged. By contrast, concatenating the Lüders instrument with an outcome-dependent bit flip shifts the ring further towards the north pole, resulting in a higher average fidelity between the input and output states.}
    \label{fig:latitude_bloch_sphere}
\end{figure}
Equivalently, these are the pure states satisfying $\bra{\psi} \sigma_z \ket{\psi} = \cos(\theta)$. Unlike the Haar ensemble, latitude ensembles encode prior information about the likely input states. This additional structure allows the post-measurement state update to be tailored to the ensemble, potentially increasing the average fidelity.

To compute the statistical disturbance bound with respect to the ensemble $\mathcal{E}_{\theta}$, we first notice that the ensemble operator $R(\mathcal{E}_{\theta})$, see Eq.~\eqref{eq:ensemble_operator}, is expressed in the computational basis via 
\begin{align}           
    &R(\mathcal{E}_{\theta}) = 2[\cos^4(\theta/2) \ketbra{00} + \sin^4(\theta/2)\ketbra{11} \nonumber \\
    &+ (\sin(\theta/2)\cos(\theta/2))^2(\ketbra{01}+\ketbra{10} \nonumber\\
    &+ \ketbra{00}{11} + \ketbra{11}{00})], \label{eq:latitude_ensemble_operator}
\end{align}
see Appendix \ref{app:latitude} for the derivation.
Using Eq.~\eqref{eq:average_fidelity_choi} together with Eq.~\eqref{eq:latitude_ensemble_operator}, the average fidelity of the Lüders channel \eqref{eq:ad_lüders_channel} with respect to the latitude ensemble becomes
\begin{equation*}
     F_{\mathcal{E}_{\theta}}(\mathcal{B}_{L}^{\mu}) = \frac{1}{2}(1+ \cos^2(\theta) + \sqrt{1-\mu} \sin^2(\theta)),
\end{equation*}
see Figure \ref{fig:latitude_ensemble_amplitude_damping}. 
We now compare this with the instrument $\bm{\mathcal{I}}^{\mu, \text{flip}}$ obtained by applying an outcome-dependent bit flip whenever the outcome $b=1$ is observed.
Its instrument elements are
\begin{align}
    \mathcal{I}_{0}^{\mu,\text{flip}}(\rho) &= \sqrt{B_{0}^{\mu}} \rho \sqrt{B_{0}^{\mu}} \nonumber\\
    \mathcal{I}_{1}^{\mu,\text{flip}}(\rho) &=  \sigma_{x} \sqrt{B_{1}^{\mu}} \rho \sqrt{B_{1}^{\mu}} \sigma_x. \nonumber
\end{align}
For the corresponding channel $\mathcal{B}_{\text{flip}}^{\mu} = \sum_b \mathcal{I}_{b}^{\mu,\text{flip}}(\rho)$ the average fidelity is computed as
\begin{align}
    F_{\mathcal{E}_{\theta}}(\mathcal{B}_{\text{flip}}^{\mu}) &= \frac{1}{2}(1+ \cos^2(\theta)(1-\mu) \nonumber  \\ &+ \cos(\theta)\mu + \sqrt{1-\mu} \sin^2(\theta)).\label{eq:flip_fidelity}
\end{align}
Since $\cos(\theta) > \cos(\theta)^2$ for $\theta \in (0,\pi/2)$, Eq.~\eqref{eq:flip_fidelity} immediately implies $F_{\mathcal{E}_{\theta}}(\mathcal{B}_{\text{flip}}^{\mu}) > F_{\mathcal{E}_{\theta}}(\mathcal{B}_{L}^{\mu})$ in this range of polar angles. Hence, for latitude ensembles in the northern hemisphere, the Lüders instrument is no longer the least disturbing realisation of the POVM.

\begin{figure}
    \centering
    \includegraphics[width=\linewidth]{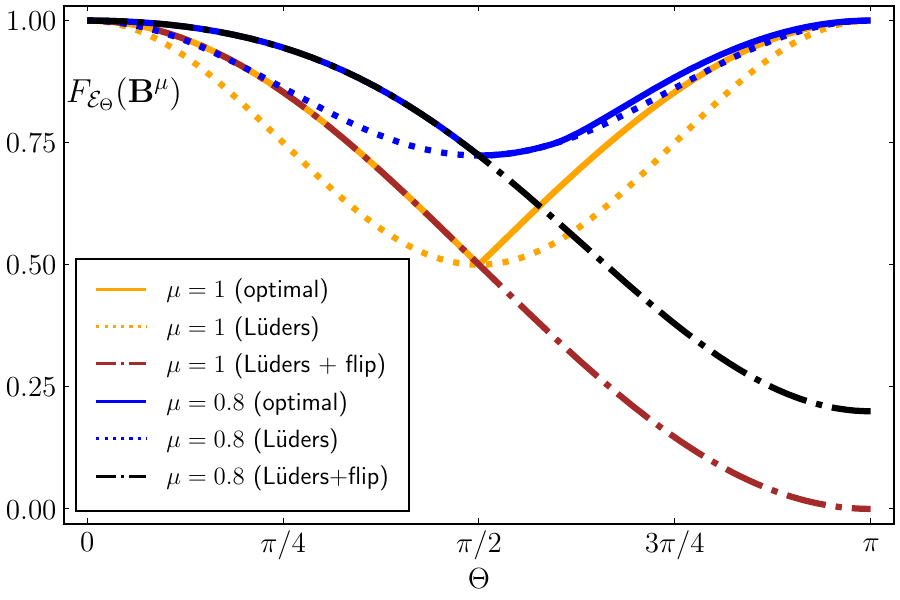}
    \caption{Statistical disturbance bounds for latitude ensembles. The statistical disturbance bound of the amplitude-damped measurement in Eq.~\eqref{eq:amplitude_damping_measurement} is shown as a function of the polar angle $\theta$. For latitude ensembles in the northern hemisphere, the Lüders instrument is no longer optimal. Instead, the instrument augmented with an outcome-dependent bit flip achieves the optimal average fidelity.}
\label{fig:latitude_ensemble_amplitude_damping}
\end{figure}

Figure~\ref{fig:latitude_ensemble_amplitude_damping} confirms that the Lüders instrument fails to achieve the statistical disturbance bound over a broad range of latitude ensembles.
In particular, throughout the northern hemisphere ($\theta\in[0,\pi/2]$), the optimal value is attained by the bit-flip instrument described by Eq.~\eqref{eq:flip_fidelity}.

The physical origin of this behaviour is simple. Whenever the outcome $b=1$ occurs, the Lüders instrument prepares the state $\ketbra{1}$. For ensembles concentrated in the northern hemisphere, however, the state $\ketbra{0}$ is, on average, considerably closer to the input states than $\ketbra{1}$. Applying an outcome-dependent bit flip therefore increases the average input-output fidelity, demonstrating that the Lüders instrument is not universally optimal outside the setting of Haar distributed input states. 

\section{Disturbance bounds in quantum randomness generation}
\label{sec:randomness}

In the previous sections, we studied how the observed measurement statistics constrain the disturbance of any compatible measurement channel with respect to a given input ensemble. It is equally natural to consider the converse perspective and to ask what can be inferred about an unknown measurement process if its disturbance is known to be small.
Questions of this type arise naturally in quantum random number generation \cite{Mannalath2023} and quantum randomness extraction \cite{Berta2014}, where the observed disturbance induced by an unknown device can be used to limit the information available to a potential eavesdropper.
In these protocols, Eve's knowledge about the generated random string is commonly quantified by her guessing probability $p_g$.
Depending on the assumptions about the preparation and measurement devices, a variety of trust models have been studied \cite{Law2014}.
To illustrate how disturbance bounds can be related to randomness certification, we consider the following simple fully trusted prepare-and-measure protocol.
\begin{enumerate}
    \item Alice prepares states drawn uniformly from the ensemble $\mathcal{E} = \{\ket{i}\}_{i=0}^{d-1}$ and sends them to Bob through an uncharacterised quantum channel $\mathcal{C}$.
    \item Bob generates a random bit $c \in \{0,1\}$. If $c=0$, then he measures in the basis $\{\ket{i}\}_{i=0}^{d-1}$ and, if $c=1$, he performs the POVM $\{A_a\}_{a=0}^{m-1}$.
    \item Alice and Bob compare their preparations and measurement settings and outcomes. 
    \item Using the rounds with $c=0$, Alice and Bob estimate a lower bound
    \begin{equation*}
        f \leq F_{\mathcal{E}}(\mathcal{C}) = \frac{1}{d} \sum_{i=0}^{d-1} \bra{i} \mathcal{C}(\ketbra{i})\ket{i}.
    \end{equation*}
    on the average fidelity of the channel with respect to the computational basis ensemble.
    \item If $c=1$, the outcomes $a \in \{0,\dots,m-1\}$ of Bob's measurement are used to generate the random string.
\end{enumerate}
The observed fidelity bound $f$ restricts the class of channels compatible with the experimental data and therefore limits the strategies available to an eavesdropper.
Since the dimension $d$ of Alice's and Bob's apparata are fixed, Eve's most general guessing strategy is described as a quantum instrument $\bm{\mathcal{I}}$ with elements $\mathcal{I}_a:\Herm{d} \rightarrow \Herm{d}$ and outcomes $a \in \{0,\dots,m-1\}$. 
Since the instrument implements the unknown channel $\mathcal C$, its elements must sum to $\mathcal C$, whose average fidelity has been certified to satisfy
$$F_{\mathcal{E}}(\mathcal{C}) \geq f.$$
Eve's optimal guessing probability is therefore obtained by solving the SDP
\begin{align}
    p_{g}(f)=\sup_{\bm{\mathcal{I}}} &\frac{1}{d}\sum_{i = 0}^{d-1} \sum_{a=0}^{m-1}  \Tr(\mathcal{I}_{a}(\ketbra{i}) A_a) \label{eq:guessing_prob}\\
     \text{s.t.} \;  &\mathcal{I}_{a} \in \text{CP}(d,d),\sum_{a} \mathcal{I}_{a}^{\dagger}(\mathds{1}_d) = \mathds{1}_d/d \nonumber \\
    &\sum_a F_{\mathcal{E}}(\mathcal{I}_a) \geq f. \nonumber
    \end{align}
Figure~\ref{fig:fidelity_vs_guessing_prob} shows the resulting guessing probability for the case where Bob generates randomness by measuring in the Fourier basis,
\begin{equation*}
    \ket{\phi_a} = \frac{1}{\sqrt{d}} \sum_{l=0}^{d-1} e^{2 \pi i a l/d} \ket{l}
\end{equation*}
for $a \in \{0,\dots,d-1\}$. 
As expected, increasing the Hilbert-space dimension reduces the fidelity required to keep Eve's guessing probability below a given threshold.

We emphasize that the above protocol merely serves as a proof of principle illustrating how disturbance bounds can be incorporated into quantum randomness generation protocols. Developing practical device-independent or semi-device-independent randomness extraction protocols based on the statistical disturbance bound is an interesting direction for future work.
\begin{figure}[t]
    \centering    \includegraphics[width=\linewidth]{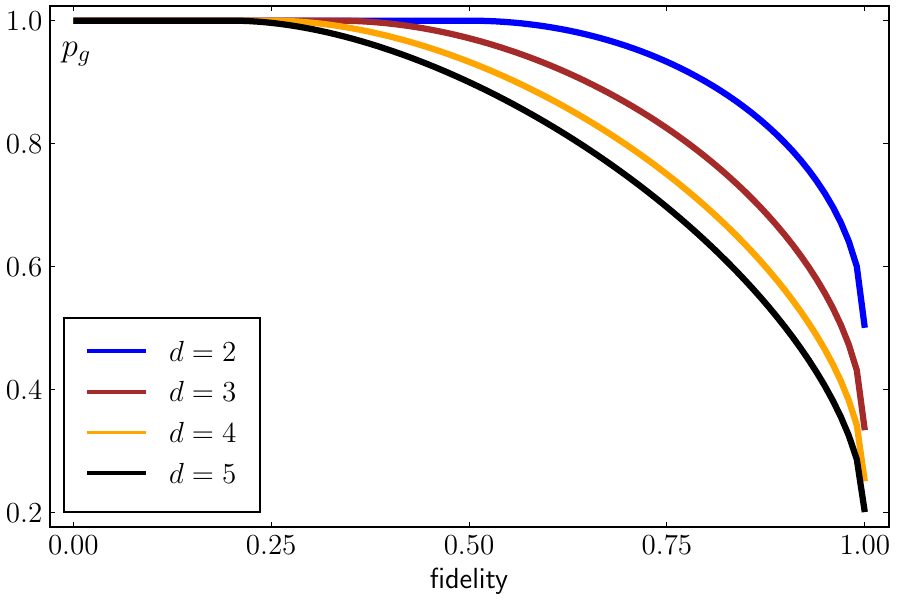}
    \caption{Guessing probability versus average fidelity. Eve's optimal guessing probability $p_g$, obtained from the SDP in Eq.~\eqref{eq:guessing_prob}, is shown as a function of the experimentally certified average fidelity with respect to the computational basis. Bob generates the random string by measuring in the Fourier basis. Higher-dimensional systems require a lower certified fidelity to achieve the same upper bound on Eve's guessing probability.}
\label{fig:fidelity_vs_guessing_prob}
\end{figure}

\section{Conclusion}

In this work, we introduced the statistical disturbance bound, which quantifies how the statistical description of a quantum measurement fundamentally limits the disturbance that any compatible measurement channel can induce with respect to an arbitrary ensemble of input states. We showed that this bound can be computed efficiently by means of semidefinite programming whenever the POVM describing the measurement is known. Moreover, we introduced the weighted state exclusion technique, which allows for an experimental estimation of the statistical disturbance bound without requiring a tomographic reconstruction of the measurement.

Our results demonstrate that the statistical disturbance bound provides substantially tighter predictions than existing information-disturbance relations and is capable of distinguishing measurements that are statistically equivalent with respect to their informativeness but differ significantly in their disturbance properties. Furthermore, we showed that the Lüders instrument is not universally the least-disturbing implementation of a measurement, but that its optimality depends on the underlying ensemble of input states. Finally, we illustrated how disturbance bounds with respect to particular state ensembles can be incorporated into a simple quantum randomness generation protocol, where the observed disturbance yields an upper bound on the guessing probability of a potential eavesdropper.

Several interesting directions remain for future research. An important question is how the framework developed in this work can be employed to design quantum measurements that simultaneously certify desired physical properties while introducing as little disturbance as possible. It would also be interesting to investigate statistical disturbance bounds for more general classes of measurements and state ensembles, as well as to explore their role in practical quantum communication and cryptographic protocols. Finally, the here developed connection between measurement disturbance and certain state exclusion tasks may offer a novel approach for studying measurement disturbance in continuous-variable systems, in analogy to the recently developed techniques to quantify resources in CV systems via quantum information tasks \cite{Haapasalo2021, Regula2021, kuramochi2020}.   

\section*{Acknowledgements}
We would like thank Paul Skrzypczyk for fruitful discussions.
We acknowledge the support from the Swedish Research Council [grant number 2024-05341] and the Wallenberg Initiative on Networks and Quantum Information (WINQ).

\section*{Code availabilty}
The source code for the implementation of the optimisation problems and the scripts used to generate the figures presented in this work are publicly available through our GitLab repository~\cite{github_repo}.
\newpage

\appendix

\section{Proof of Theorem \ref{thm:lüders_optimality}}
\label{app:lüders_optimality}

In this appendix, we prove the optimality of the Lüders instrument for disturbance minimisation with respect to the Haar ensemble, as stated in Theorem~\ref{thm:lüders_optimality}. The proof relies on the following technical lemma, which we establish first.
\begin{Lemma}
    \label{lemma:trace_inverse}
    Let $Y \in \textnormal{Herm}(d)$ be a hermitian matrix and let $\Phi^{+} = \ketbra{\phi^+} =\frac{1}{d}\sum_{i,j=0}^{d-1} \ketbra{i}{j} \otimes \ketbra{i}{j} \in \textnormal{Herm}(d^2)$ be the maximally entangled state. Then, it holds that 
    \begin{equation}
        Y \otimes \mathds{1}_d \succcurlyeq \Phi^{+} \label{eq:harmonic_phi_condition}
    \end{equation}
    if and only if 
    \begin{equation}
        \label{eq:harmonic_inverse_condition}
        Y \succ 0 \textnormal{ and } \Tr(Y^{-1}) \leq d.
    \end{equation}
\end{Lemma}

\begin{proof}
    We first prove the implication \eqref{eq:harmonic_phi_condition} $\Rightarrow$ \eqref{eq:harmonic_inverse_condition}
    and start by showing that $Y \succ 0$, i.e, that $Y$ can only have positive eigenvalues. We proceed by contradiction. Suppose that $Y$ admits a spectral decomposition
    \begin{equation*}
        Y = \sum_{i=0}^{d-1} y_i \ketbra{\psi_i}
    \end{equation*}
    such that $y_i \leq 0$ for some $i$. Then it also holds true that
    \begin{equation*}
        \bra{\psi_i, l} Y \otimes \mathds{1}_{d} \ket{\psi_i,l} \leq 0 
    \end{equation*}
    for all $l \in \{0,\dots,d-1\}$, where we used the standard notation $\ket{\psi_i,l} := \ket{\psi_i} \otimes \ket{l}$. On the other hand, one has that
    \begin{equation*}
        \bra{\psi_i, l} \Phi^+ \ket{\psi_i,l} = \frac{1}{d} \abs{\braket{\psi_i} {l}}^2,
    \end{equation*}
    which is strictly positive for at least one $l \in \{0,\dots,d-1\}$ by the completeness of the orthonormal basis $\{\ket{i}\}$. This contradicts the assumed inequality \eqref{eq:harmonic_phi_condition}.
    
   We next show that $\Tr(Y^{-1}) \leq d$. To see this, we consider the vector $\ket{v} = Y^{-1} \otimes \mathds{1}_d \ket{\phi^{+}}$ and calculate 
    \begin{align}
        &\bra{v} Y\otimes \mathds{1}_{d} - \Phi^+ \ket{v}  \nonumber\\  
        &= \frac{1}{d} \Tr(Y^{-1}) - \left(\frac{1}{d} \Tr(Y^{-1})\right)^2 \nonumber
    \end{align}
    which is nonnegative only if $\Tr(Y^{-1}) \leq d$. Here, we have used the well-known identity
    \begin{equation}
        \label{eq:max_ent_identity}
        \bra{\phi^+} A \otimes B \ket{\phi^+} = \frac{1}{d}\tr(A^T B).
    \end{equation}
    This completes the proof of the implication \eqref{eq:harmonic_phi_condition} $\Rightarrow$ \eqref{eq:harmonic_inverse_condition}.

    To prove the opposite direction  \eqref{eq:harmonic_inverse_condition} $\rightarrow$ \eqref{eq:harmonic_phi_condition}, we first note that the inequality in \eqref{eq:harmonic_inverse_condition} is, using the identity \eqref{eq:max_ent_identity}, equivalent to 
    \begin{equation}
        \label{eq:trace_inverse_phi}
        \bra{\phi^+} Y^{-1} \otimes \mathds{1}_d \ket{\phi^+} \leq 1. 
    \end{equation}
    Next, let $\ket{u} \in \mathbb{C}^d \otimes \mathbb{C}^d$ be an arbitrary vector.
    Multiplying both sides of \eqref{eq:trace_inverse_phi} by
 $\braket{u}$ and rearranging the terms yields
    \begin{align}
          0 &\leq \braket{u}{u} - \braket{u}{u} \bra{\phi^+} Y^{-1} \otimes \mathds{1}_d \ket{\phi^+} \nonumber\\
          &\leq \braket{u}{u} - \abs{\bra{u} Y^{-1/2} \otimes \mathds{1}_d \ket{\phi^+}}^2, \nonumber
    \end{align}
    where the second inequality from the Cauchy-Schwarz inequality. Since this inequality holds for every vector $\ket{u}$, it is equivalent to the operator inequality
    \begin{equation*}
        \mathds{1}_{d^2} \succcurlyeq (Y^{-1/2} \otimes \mathds{1}_d) \Phi^{+}  (Y^{-1/2} \otimes \mathds{1}_d),
    \end{equation*}
    which further implies the operator inequality \eqref{eq:harmonic_phi_condition}.
    \end{proof}
    
We are now in a position to prove Theorem \ref{thm:lüders_optimality} that stated the following.
\begin{theoremcopy}
    Let $\bm{A} = \{A_{a}\}_{a=0}^{m-1} \in \Herm{d}$ be a $d$-dimensional POVM. 
    Then, the statistical disturbance bound $F_{\mathcal{H}}(\bm{A})$ with respect to the Haar ensemble $\mathcal{H}$ is given by 
    \begin{equation}
        \label{eq:lüders_value_appendix}
        F_{\mathcal{H}}(\bm{A}) = \frac{\frac{1}{d} \left[\sum_{a=0}^{m-1} \Tr (\sqrt{A_a})^2\right]+1}{d+1}.
    \end{equation}
\end{theoremcopy}
\begin{proof}
    The value on the right-hand side of Eq.~\eqref{eq:lüders_value_appendix} is attained by the Lüders instrument $\mathcal{I}_{a}(\rho) = \sqrt{A_a} \rho \sqrt{A_a}$, and therefore provides a lower bound on $F_{\mathcal{H}}(\bm{A})$. To show that this lower bound is optimal, we use the dual semidefinite program in Eq.~\eqref{eq:dual_fidelity_sdp}. For a single measurement effect $A_a \in \bm{A}$, it reads
    \begin{align}
    &F_{\mathcal{H}}(A_a) = \inf_{Y_a}  \frac{1}{d}  \Tr(Y_a A_a)  \label{eq:dual_feasible_haar_objective_app}\\
    &\text{s.t. }   Y_a^T \otimes \mathds{1}_d \succcurlyeq R(\mathcal{\mathcal{H}}) = \frac{\mathds{1}_{d^2} + d\ketbra{\phi^{+}}}{d+1}. \label{eq:dual_feasible_haar_constrained_app}
    \end{align}
    We next construct a dual feasible solution attaining the value
    \begin{equation*}
        F_a^{*} = \frac{1}{d(d+1)}\left(\Tr(A_a) + \Tr(\sqrt{A_a})^2 \right).
    \end{equation*}
    Since $\sum_{a} F_a^{*}$ coincides with the right-hand side of Eq.~\eqref{eq:lüders_value_appendix}, weak duality implies that proving dual feasibility establishes optimality. 
    To show that $F_a^{*}$ can be attained, we consider the family of operators $Y_{a,r}$ defined by 
    \begin{equation*}
        Y_{a,r} = \frac{\mathds{1}_d}{d+1} + \frac{\Tr(\sqrt{A_a})A_a^{-1/2}}{d+1-q_a/r} + rQ_a
    \end{equation*}
    where $Q_{a}$ is the orthogonal projector onto the kernel of $A_a$, $q_a:=\Tr(Q_{a})$ is the dimension of the kernel of $A_a$, $A_{a}^{-1/2}$ is the Moore-Penrose pseudo-inverse on the support of $A_a^{1/2}$ and $r$ is any number such that $r > \frac{q_a}{d+1}$. In the case that $A_a$ has full rank, the pseudo-inverse is the usual inverse and $Q_a = 0$. Since $Q_a$ projects onto the kernel of $A_a$, one has $\Tr(A_a Q_a) = 0$ and hence we obtain
    \begin{equation*}
        \frac{1}{d}\Tr(Y_{a,r}A_{a}) = \frac{1}{d(d+1)}\left[\Tr(A_a) + \frac{\Tr(\sqrt{A_a})^2}{1-\frac{q_a}{r(d+1)}}\right]
    \end{equation*}
    which converges monotonically to $F_a^{*}$ from above as $r \rightarrow \infty$. It remains to verify that $Y_{a,r}$ satisfies the dual constraint \eqref{eq:dual_feasible_haar_constrained_app} for all $r > \frac{q_a}{d+1}$. Hence, by rearranging the terms in the inequality \eqref{eq:dual_feasible_haar_constrained_app}, it remains to be shown that 
    \begin{equation*}
        \left(\frac{d+1}{d}\right) \left(Y_{a,r}^{T} - \frac{\mathds{1}_d}{d+1}\right) \otimes \mathds{1}_d \succcurlyeq \ketbra{\phi^{+}}. 
    \end{equation*}
    Since $Y_{a,r}$ is positive definite, Lemma~\ref{lemma:trace_inverse} shows that this operator inequality is equivalent to
    \begin{align}
        &\frac{d}{d+1}\Tr[\left(Y_{a,r} - \frac{\mathds{1}_{d}}{d+1}\right)^{-1}] \leq d. \label{eq:y_a,r_inverse_ineq}
    \end{align}
   A direct calculation shows that the left-hand side of \eqref{eq:y_a,r_inverse_ineq} is equal to $d$, and hence the inequality is saturated. Consequently, $\sum_{a} F_a^{*}$ is both a lower and an upper bound for $F_{\mathcal{H}}(\bm{A})$ which completes the proof.
\end{proof}

\section{Derivation of the ensemble operator of the latitude ensemble}
\label{app:latitude}
In this appendix, we derive Eq.~\eqref{eq:latitude_ensemble_operator} for the ensemble operator $R(\mathcal{E}_{\theta})$ associated with the latitude ensemble containing the pure qubit states that fulfill $\bra{\psi}\sigma_z  \ket{\psi} = \cos(\theta)$ for a fixed polar angle $\theta \in [0,\pi]$. 
The states in this ensemble are precisely the pure qubit states $\ketbra{\theta, \phi}$ given by 
\begin{align}
    \begin{pmatrix}
        \cos^2(\theta/2) & \sin(\theta/2) \cos(\theta/2) e^{-i\phi} \\
        \sin(\theta/2) \cos(\theta/2) e^{i\phi} & \sin^2(\theta/2)
    \end{pmatrix}
    \nonumber
\end{align}
with an arbitrary azimuthal angle $\phi \in [0,2\pi]$.

To compute the ensemble operator $R(\mathcal{E}_{\theta})$, we evaluate
\begin{equation}
    \label{eq:latitude_integral}
    R(\mathcal{E}_{\theta}) = 2\int d\mu_{\mathcal{E}_{\theta}} \ketbra{\theta',\phi}^T \otimes  \ketbra{\theta',\phi}
\end{equation}
where the probability measure $\mu_{\mathcal{E}_{\theta}}$ is defined via
\begin{equation}
    \label{eq:latitude_measure_integration}
    \int d\mu_{\mathcal{E}_{\theta}}\, f(\theta', \phi)  = \frac{1}{2\pi} \int_{0}^{2\pi} d\phi \, f(\theta,\phi) 
\end{equation}
for any function $f(\theta', \phi)$ on the Bloch sphere. Many matrix elements $\bra{ij} R(\mathcal{E}_{\theta})\ket{kl}$ vanish due to the identity 
\begin{equation*}
    \frac{1}{2\pi} \int_{0}^{2\pi} d\phi \, e^{i k \phi} = \delta_{k0}.
\end{equation*}
Evaluating the integral in Eq.~\eqref{eq:latitude_integral} then yields
\begin{align}           
    &R(\mathcal{E}_{\theta}) = 2[\cos^4(\theta/2) \ketbra{00} + \sin^4(\theta/2)\ketbra{11} \nonumber \\
    &+ (\sin(\theta/2)\cos(\theta/2))^2(\ketbra{01}+\ketbra{10} \nonumber\\
    &+ \ketbra{00}{11} + \ketbra{11}{00})].
\end{align}

\printbibliography

\end{document}